\documentclass{rQUF2e}
\usepackage[english]{babel}
\usepackage{amssymb}
\usepackage{amsfonts}

\usepackage[dvips]{color}
\usepackage{graphicx}
\usepackage{amsmath}
\usepackage{dsfont}
\usepackage{tikz}

\usepackage{bbm}

\usepackage{float}

\usepackage{graphics}
\usepackage{graphicx}
\usepackage{amsmath}
\usepackage{rotating}
\usepackage{varioref}
\usepackage{multirow}
\usepackage{pdflscape}
\usepackage{natbib}

\theoremstyle{plain}
\theoremstyle{definition}
\theoremstyle{remark}
\newtheorem{theorem}{Theorem}

\newtheorem{assumption}[theorem]{Assumption}

\newtheorem{corollary}[theorem]{Corollary}

\newtheorem{definition}[theorem]{Definition}
\newtheorem{example}[theorem]{Example}

\newtheorem{lemma}[theorem]{Lemma}

\newtheorem{proposition}[theorem]{Proposition}
\newtheorem{remark}[theorem]{Remark}

\begin{document}

\title{{\textit{On the properties of the Lambda value at risk: robustness, elicitability and consistency}}}

\author{M. BURZONI$^{\ast}$\thanks{$^\ast$Email: matteo.burzoni@math.ethz.ch}, I. PERI$^{\dag}$\thanks{$^\dag$Corresponding author.
Email: i.peri@greenwich.ac.uk} and C. M. RUFFO$^{\ddag}$\thanks{$^\ddag$Email: c.ruffo@campus.unimib.it}\\
\affil{$\ast$Department of Mathematics, ETH Z\"{u}rich, R\"{a}mistrasse 101, 8092 Z\"{u}rich, Switzerland\\
$\dag$Department of Finance, University of Greenwich, 30 Park Row, London SE10 9LS, UK\\
$\ddag$Department of Statistics and Quantitative Methods, University of Milano Bicocca, Via Bicocca degli Arcimboldi 8, 20126 Milan, Italy} \received{v1 released April 2016} }

\maketitle

\begin{abstract}
Recently, financial industry and regulators have enhanced the debate on the good properties of a risk measure. A fundamental issue is the evaluation of the quality of a risk estimation. On the one hand, a backtesting procedure is desirable for assessing the accuracy of such an estimation and this can be naturally achieved by elicitable risk measures. For the same objective, an alternative approach has been introduced by \cite{Davis} through the so-called consistency property. On the other hand, a risk estimation should be less sensitive with respect to small changes in the available data set and exhibit qualitative robustness. A new risk measure, the Lambda value at risk ($\Lambda VaR$), has been recently proposed by \cite{FMP14}, as a generalization of $VaR$ with the ability to discriminate the risk among P$\&$L distributions with different tail behaviour. In this article, we show that $\Lambda VaR$ also satisfies the properties of robustness, elicitability and consistency under some conditions.
\end{abstract}


\begin{classcode}C13, D81, G17\end{classcode}


\section{Introduction}
Risk measurement is a matter of primary concern to the financial services industry. The most widely used risk measure is the value at risk ($VaR$), which is the negative of the right $\lambda$-quantile $q^+_{\lambda}$, for some conventional confidence level $\lambda$ (e.g. $1 \%$). $VaR$ became popular as a law invariant risk measure for its simple formulation and facility of computation, however, it presents several limits. First, $VaR$ lacks convexity with respect to random variables which, in general, penalize diversification. $VaR$ satisfies, instead, the quasi-convexity property with respect to distributions \citep{Drapeau,FMP14}. This condition has a natural interpretation in terms of compound lotteries: the risk of the compound lottery is not higher than the one of the riskiest lottery. Another relevant issue of $VaR$ is the lack of sensitivity to the tail risk as it attributes the same risk to distributions having the same quantile but different tail behaviour.

 Recently, a new risk measure, the Lambda value at risk ($\Lambda VaR$), has been proposed by \cite{FMP14}. $\Lambda VaR$ seems to be interesting for its ability to capture the tail risk by generalizing $VaR$. Specifically, $\Lambda VaR$ is defined as follows:
 \begin{equation*}
\Lambda VaR(F):=-\inf\{x \in \mathbb{R} : F(x) > \Lambda(x)\}
\end{equation*}
where $\Lambda:\mathbb{R} \rightarrow [\lambda^m, \lambda^M]$ with $0<\lambda^m \leq \lambda^M <1$ is a right continuous and monotone function. When the $\Lambda$ function is constantly equal to some $\lambda\in (0,1)$ it coincides with the definition of $VaR$ with confidence level $\lambda$. The main idea is that the confidence level can change and it is a function of the asset's losses. In this way, $\Lambda VaR$ is able to discriminate the risk among P$\&$L distributions with the same quantile but different tail behaviour. In this regard, the sensitivity of $\Lambda VaR$ is up to the $\lambda^m$-quantile of a distribution, since, by definition, $\Lambda VaR(F)\leq VaR_{\lambda^m}(F)$. Nevertheless, the requirement $\lambda^m>0$ is only technical and $\lambda^m$ can be chosen arbitrarily close to $0$.
Properties of $\Lambda VaR$ such as monotonicity and quasiconvexity are obtained in \citep{FMP14} in full generality (i.e. allowing also for $\lambda^m=0$). 

The purpose of this paper is to study if $\Lambda VaR$ satisfies other important properties for a risk measure also satisfied by  $VaR$. We first focus on the so-called \emph{robustness} that refers to the insensitivity of a risk estimator to small changes in the data set. We adopt the Hampel's classical notion of qualitative robustness \citep{Hampel,Huber}, also considered by \cite{CDS10} for general risk measures (a stronger notion has been later proposed by \cite{K12a} for convex risk measures). We show that the historical estimator of $\Lambda VaR$ is robust within a family of distributions which depends on $\Lambda$. In particular, we recover the result of \cite{CDS10} for $VaR$, in the case of $\Lambda\equiv \lambda\in(0,1)$.

A second property we investigate is the \emph{elicitability} for $\Lambda VaR$ . Several authors underlined the importance of this property in the risk management and backtesting practice \citep{G11,Z14,EH14,BB14}. Specifically, the elicitability allows the comparison of risk measure forecasts and provides a natural methodology to perform the backtesting. As for the case of $VaR$, also  $\Lambda VaR$ is elicitable in a particular family of distributions which depends on $\Lambda$ and, for the particular case of $\Lambda\equiv \lambda\in(0,1)$, we recover the results of \cite{G11}. Note that the elicitability for $\Lambda$ decreasing was already observed by \cite{BB14}, we extend here to the most interesting case of $\Lambda$ increasing.

Finally, we study the consistency property, as recently proposed by \cite{Davis}\footnote{During the review process of this paper the consistency property has been renamed by Davis as \textit{calibration} of predictions in a dynamic setting.}. In this study, Davis argues that the decision-theoretic framework of elicitability assumes the strong assumption that the theoretical P$\&$L distribution is known and remain unchanged at any time. He thus suggests, under a more refined framework, the use of the so-called \emph{consistency} property, in order to verify if a risk measure produces accurate estimates.We show that $\Lambda VaR$ satisfies the consistency property without any assumption on the P$\&$L generating process, as in the case of $VaR$.

The structure of the paper is as follows. After introducing the basic notions and definitions, in Section \ref{Notations and definitions}, we start examining the robustness property in Section \ref{robustness}. We dedicate the Section \ref{elicitability} to the elicitability of  $\Lambda VaR$. Finally, in Section \ref{consistency}, we refine the theoretical framework and we verify the consistency of $\Lambda VaR$. 

\section{Notations and definitions}
\label{Notations and definitions}
Let $(\Omega ,\mathcal{F},\mathbb{P})$ be a non-atomic probability space and $L^{0}:=L^{0}(\Omega ,\mathcal{F},\mathbb{P})$ be the space of $\mathcal{F}$-measurable random variables that are $\mathbb{P}$-almost surely finite. We assume that $X \in L^0$ represents a financial position (i.e. a loss when $X<0$ and a profit when $X>0$). Any random variable $X\in L^{0}$ induces a probability measure $P_{X}$ on $(\mathbb{R}$,$\mathcal{B}_{\mathbb{R}})$ by $P_{X}(B)=\mathbb{P}(X^{-1}(B))$ for every Borel set $B\in \mathcal{B}_{\mathbb{R}}$ and $F(x):=P_{X}(-\infty ,x]$ denote its distribution function. Let $\mathcal{D}:=\mathcal{D}(\mathbb{R})$ be the set of distribution functions and $\mathcal{D}_1$ those with finite first moment. 

A risk measure is a map $\rho:L \subseteq L^0 \rightarrow \overline{\mathbb{R}}$ that assigns to each return $X\in L$ a number representing the minimal amount of capital required by the regulator in order to cover its financial risk. The majority of risk measures used in finance are distribution-based risk measures, that is, they assign the same value to random variables with the same distribution. Such risk measures $\rho$ are called law-invariant, more formally they satisfy:
\begin{equation*}
X\sim _{d}Y\Rightarrow \rho (X)=\rho (Y).
\end{equation*}
In this way, a risk measure $\rho$ can be represented as a map on a set $\mathcal{M}\subseteq \mathcal{D}$ of distributions. With a slight abuse of notation, we still denote this map by $\rho$ and set:  
\begin{equation*}
\rho(F) := \rho(X)
\end{equation*}
where $F$ is the distribution function of $X$. Since the seminal paper by \cite{Artzner}, the theory of risk measures has been based on the study of their minimal properties. Also when risk measures are defined on distributions, monotonicity is generally accepted; formally, for any $F_1,F_2 \in \mathcal{M}$, $\rho$ is monotone if:
\begin{equation*}
F_1(x) \geq F_2(x),\  \forall x \in \mathbb{R} \ \text{ implies } \ \rho(F_1) \leq \rho(F_2).
\end{equation*}
Other properties have been discussed by academics. As pointed out in \cite{FMP14}, the convexity property, for risk measures defined on distributions, is not compatible with the translation invariance property. Thus, we might require $\rho$ to satisfy quasiconvexity \citep{Drapeau,FMP14}:
\begin{equation*}
\text{for any } \gamma \in [0,1], \ \ \rho(\gamma F_1 + (1-\gamma)F_2) \leq \max(\rho(F_1),\rho(F_2)).
\end{equation*}
 
It is widely accepted in the financial industry to adopt the risk measure Value at Risk ($VaR$) at a confidence level $\lambda \in (0,1)$, that is defined as follows \citep[see][Definition 3.3]{Artzner}: 
\begin{equation}
VaR_{\lambda}(F):=-\inf\{x \in \mathbb{R} : F(x) > \lambda\}.
\label{VaR}
\end{equation}
 $VaR$ is monotone and quasiconvex \citep{FMP14} but, obviously from the definition, it is not tail-sensitive. In order to overcome its limits, the \cite{FRTB} recommends the use of Expected Shortfall ($ES$), formally given by:
\begin{equation}
ES_{\lambda}(F):=\frac{1}{\lambda}\int^{\lambda}_{0} VaR_{s}(F)ds.
\label{eqES}
\end{equation}

$ES$ is able, by definition, to evaluate the tail risk  and it satisfies the subadditivity property on random variables \citep{Artzner}. 

Another tail sensitive risk measure is Lambda Value at Risk ($\Lambda VaR$), recently introduced by \cite{FMP14}, whose properties are the main topic of this paper. $\Lambda VaR$ generalizes $VaR$ by considering a function $\Lambda$ instead of a constant $\lambda$ in the definition of $VaR$. The advantages of considering the $\Lambda$ function are twofold: on the one hand, $\Lambda VaR$ provides a criterion to change the confidence level when the market condition changes (e.g. putting aside more capital in case of expected greater losses), on the other hand, it allows differentiating the risk of P$\&$L distributions with different tail behaviour. Formally, $\Lambda VaR$ is defined by:
\begin{definition}\label{defLVaR}
 \begin{equation}
\Lambda VaR(F):=-\inf\{x \in \mathbb{R} : F(x) > \Lambda(x)\}
\label{LVaR}
\end{equation}
where $\Lambda:\mathbb{R} \rightarrow [\lambda^m, \lambda^M]$ with $0<\lambda^m \leq \lambda^M <1$ is a right continuous and monotone function.
\end{definition}
Intuitively, if both $F$ and $\Lambda$ are continuous,  $\Lambda VaR$ is given by the smallest intersection between $F$ and $\Lambda$. Unlike ES, $\Lambda VaR$ lacks subadditivity, positive homogeneity and translation invariance when defined on random variables, nevertheless, $\Lambda VaR$ is monotone and quasiconvex on the set of distributions \citep[for a discussion on these properties see Section 4 of][]{FMP14}.

\section{Robustness}
\label{robustness}

Evaluating the goodness of a risk measure involves determining how its computation can be affected by estimation issues. The problem consists in examining the sensitivity of a risk measure to small changes in the available data set; for this reason, \textit{robustness} seems to be a key property. In this context, the first rigorous study is given by \cite{CDS10}. The authors pointed out that the notion of robustness should be referred to the ``risk estimator'', as outcome of a ``risk measurement procedure'' \citep[see][for details]{CDS10}, and they founded the problem on the Hampel's classical notion of qualitative robustness \citep{Hampel,Huber}. Basically, a risk estimator is called robust if small changes in the P$\&$L distribution implies small changes in the law of the estimator. They consider the case of historical estimators $\hat{\rho}^h$, those obtained by applying the risk measure $\rho$ to the empirical distribution $\hat{F}$, and they conclude that historical estimator of $VaR$ leads to more robust procedures than alternative law-invariant coherent risk measures. 

Afterwards, \cite{K12b} and \cite{K12a} argued that the Hampel's notion does not discriminate among P$\&$L distributions with different tail behaviour and, hence, is not suitable for studying the robustness of risk measures that are sensitive to the tails, such as $ES$. So they focused on the case of law-invariant coherent risk measures and they showed that robustness is not entirely lost, but only to some degree, if a stronger notion is used. 

Substantially, the robustness of a risk estimator is based on the choice of a particular metric and different metrics leads to a more or less strong definition. However, as pointed out by \cite{E14}, a proper definition of robustness is still a matter of primary concern. The aim of this section is to study the robustness of $\Lambda VaR$, where we use the weakest definition of robustness proposed by \cite{CDS10}. 

Let us denote with $\textbf{x}\in \mathcal{X}$ the $n$-tuple representing a particular data set, where $\mathcal{X}=\cup_{n \geq 1} \mathbb{R}^n$ is the set of all the possible data sets. The estimation of $F$ given a particular data set $\textbf{x}$ is denoted with $\hat{F}$ and represents the map $\hat{F}:\mathcal{X}\rightarrow \mathcal{D}$. We call risk estimator the map $\hat{\rho}:\mathcal{X} \rightarrow \mathbb{R}$ that associates to a specific data set $\textbf{x}$ the following value: $$\hat{\rho}(\textbf{x}):=\rho(\hat{F}(\textbf{x})).$$ In particular, the historical estimator $\hat{\rho}^h$ associated to a risk measure $\rho$ is the estimator obtained by applying $\rho$ to the empirical P$\&$L distribution, $F^{emp}$, defined by $F^{emp}(x):=\frac{1}{n}\sum_{i=1}^n \mathbf{1}_{{(x \geq x_i)}}$ with $n \geq 1$, that is: $$\hat{\rho}^h(\textbf{x}):=\rho(F^{emp}(\textbf{x})).$$
Let us denote with $d(\cdot,\cdot)$ the L\'evy metric, such that for any two distributions $F,G\in\mathcal{D}$ we have
\begin{equation*}
d(F,G):=\inf\{\varepsilon>0\mid F(x-\varepsilon)-\varepsilon\leq G(x)\leq F(x+\varepsilon)+\varepsilon \quad\forall x\in\mathbb{R}\}.
\end{equation*}
Hereafter, we recall the definition of $\mathcal{C}$-robustness of a risk estimator as proposed by \cite{CDS10}, where $\mathcal{C}$ is a subset of distributions. 
\begin{definition}\citep{CDS10}
A risk estimator $\hat{\rho}$ is $\mathcal{C}$-robust at $F$ if for any $\varepsilon> 0$ there exists $\delta> 0$ and $n_0>1$ such that, for all $G\in \mathcal{C}$:
$$d(G,F)\leq\delta \ \Rightarrow \ d(\mathcal{L}_n(\hat{\rho},G),\mathcal{L}_n(\hat{\rho},F))\leq\varepsilon \quad \forall n\geq n_0$$
where $d$ is the L\'evy distance and $\mathcal{L}_n(\hat{\rho},F)$ is the law of the estimator $\rho(\hat{F}(\mathbf{X}))$ with $\mathbf{X}:=(X_1,\ldots, X_n)$ a vector of independent random variables with common distribution $F$.
\end{definition}

As a consequence of a generalization of the Hampel's theorem, \cite{CDS10} obtained the following result: 
\begin{corollary}\citep{CDS10}
If a risk measure $\rho$ is continuous in $\mathcal{C}$ respect to the L\'evy metric, then the historical estimator, $\hat{\rho}^h$ is $\mathcal{C}$-robust at any $F\in\mathcal{C}$ .
\label{contcorol}
\end{corollary}
Hence, they show that the historical estimator of $VaR_\lambda$ is robust with respect to the following set:
 \begin{equation}
 \mathcal{C}_{\lambda}:=\left\{F\in\mathcal{D}\mid q^-_\lambda(F)=q^+_\lambda(F)\right\}
\label{Clambda}
 \end{equation}
 where $q^{+}_\lambda(F):=\inf \left\{x \mid F(x)> \lambda \right\}$ and  $q^{-}_\lambda(F):=\inf \left\{x \mid F(x)\geq \lambda \right\}$. Substantially, when the quantile of the true P$\&$L distribution is unique, then the empirical quantile is robust. In addition, they showed that the historical estimator of  $ES_\lambda$ is not robust. More important, they pointed out a conflict between convexity (on random variables) and robustness: any time the convexity property is required on distribution-based risk measures, its historical estimator fails to be robust.

We use the result by \cite{CDS10} in Corollary \ref{contcorol} to prove under which conditions the historical estimator of  $\Lambda VaR$ is robust.
\begin{assumption}
 In this section we assume that $\Lambda:\mathbb{R}\mapsto [\lambda^m,\lambda^M]$ is a continuous function.
\end{assumption}
First, let us consider the following set:
$$E_F:=\{x\in\mathbb{R}\mid F(x)=\Lambda(x)\text{ or } F(x^-)=\Lambda(x)\}$$ which consists of those points where the distribution $F$ (or the left-continuous version of $F$) intersects $\Lambda$. We introduce the following class $\mathcal{C}_{\Lambda}$ of distributions:
 \begin{equation}\label{CLambda}
 \mathcal{C}_{\Lambda}:=\left\{F\in\mathcal{D}\mid F((x,x+\varepsilon))> \Lambda((x,x+\varepsilon))\quad\text{for some }\varepsilon=\varepsilon(x)> 0,\ \forall x\in E_F\right\}
 \end{equation}
 where  $F((x,x+\varepsilon))$ and $\Lambda((x,x+\varepsilon)$ are the images of the interval $(x,x+\varepsilon)$ through $F$ and $\Lambda$ respectively. The set $\mathcal{C}_{\Lambda}$ consists of those distributions that do not coincide with $\Lambda$ on any interval. In the special case of $\Lambda\equiv\lambda\in(0,1)$, it simply means that the quantile is uniquely determined, thus, the family $\mathcal{C}_{\Lambda}$  coincides with the one in \eqref{Clambda} considered by \cite{CDS10} for the robustness of $VaR_\lambda$. Note also that for $\Lambda$ decreasing this condition is automatically satisfied and hence $\mathcal{C}_{\Lambda}=\mathcal{D}$.

In the following proposition we show that the historical estimator of $\Lambda VaR$ is robust in the class $\mathcal{C}_{\Lambda}$ of distribution functions.
 \begin{proposition} $\Lambda VaR$ is continuous on $\mathcal{C}_{\Lambda}$. Hence, $\widehat{\Lambda VaR}^h$ is $\mathcal{C}_{\Lambda}$-robust.
 \end{proposition}
 \begin{proof}
 We only need to show continuity of $\Lambda VaR$ respect to the L\'evy metric the rest follows from Corollary \ref{contcorol} by \cite{CDS10}.\\
Fix $\varepsilon>0$ and $F\in\mathcal{C}_{\Lambda}$. Let $\overline{x}:= -\Lambda VaR(F)$. For any $n\in\mathbb{N}$, define the sets  $A_n:=\{x\in (-\infty,\overline{x}-\varepsilon]\mid \Lambda(x)-F(x+1/(2n))\geq 1/n\}$. Observe that, for $x\in A_n$, we have $$\frac{1}{n+1}\leq \frac{1}{n}\leq \Lambda(x)-F\left(x+\frac{1}{2n}\right)\leq \Lambda(x)-F\left(x+\frac{1}{2(n+1)}\right)$$ and hence
$A_n\subseteq A_{n+1}$. We first show that $$(-\infty,\overline{x}-\varepsilon]=\bigcup_{n\in\mathbb{N}}A_n.$$
 The inclusion $\supseteq$ is obvious. Fix $x\in (-\infty,\overline{x}-\varepsilon]$ and let $\gamma:=\Lambda(x)-F(x)$. By definition of $\overline{x}$ and $\mathcal{C}_{\Lambda}$ we have that $\Lambda(x)> F(x)$ and hence $\gamma>0$. From the right-continuity of $\Lambda-F$, and the continuity of $\Lambda$, for any $\varepsilon'>0$ there exists $n_0\in\mathbb{N}$ such that $\forall n\geq n_0$, $\Lambda(x+1/(2n))-F(x+1/(2n))\geq \gamma-\varepsilon'$ and $\Lambda(x)-\Lambda(x+1/(2n))\geq -\varepsilon'$. Take now $\varepsilon'=\gamma/4$ to obtain
$$\Lambda(x)-F(x+1/(2n))=\Lambda(x+1/(2n))-F(x+1/(2n))+\Lambda(x)-\Lambda(x+1/(2n))\geq \gamma-\gamma/4-\gamma/4=\gamma/2$$
since $\gamma>0$, for a sufficiently large $n$ we get $\Lambda(x)-F(x+1/(2n))\geq 1/n$ and hence $x\in A_n$ for some $n\in\mathbb{N}$, as claimed.\\

We now show that there exists $n_0\in\mathbb{N}$ such that
$$(-\infty,\overline{x}-\varepsilon]=\bigcup_{n=1}^{n_0}A_n.$$
If indeed $A_{n+1}\setminus A_n \neq\emptyset$ for infinitely many $n\in\mathbb{N}$, then there exists a convergent subsequence $\{x_{k}\}$ with $x_{k}\in A_{n_k+1}\setminus A_{n_k}$ and $\tilde{x}:=\lim_{k\rightarrow\infty} x_k$ such that: i) $-\infty<\tilde{x}\leq\overline{x}-\varepsilon$ and ii) $F(\tilde{x})\geq \Lambda(\tilde{x})$. i) follows from the fact that $\Lambda$ has a lower bound $\lambda^m$ while $F$ obviously tends to $0$ as $x$ approaches $-\infty$. There exists therefore $M>0$ and $n_M$ such that $(-\infty,M]\subseteq A_n$ for every $n\geq n_M$; ii) follows from $x_k\notin A_{n_k}$ which implies $\Lambda(x_k)-F(x_k+1/(2n_k))< 1/{n_k}$ and the right-continuity of $F$ which implies
$$F(\tilde{x})\geq \limsup F(x_k+1/(2n_k))\geq \limsup \Lambda(x_k)-1/{n_k}=\Lambda(\tilde{x}) $$
where the last inequality follows from the continuity of $\Lambda$. If $F(\tilde{x})=\Lambda(\tilde{x})$, by definition of $\mathcal{C}_{\Lambda}$ we obtain $-\Lambda VaR(F)\leq \tilde{x}\leq\overline{x}-\varepsilon$ which is a contradiction. The same conclusion obviously follows when $F(\tilde{x})>\Lambda(\tilde{x})$.\\

We have therefore shown the existence of $n_0\in\mathbb{N}$ such that $\Lambda(x)-F(x+1/(2n_0))\geq 1/{n_0}$ for every $x\in (-\infty,\overline{x}-\varepsilon]$ Take now $\delta_1:=1/(2n_0)$ and $G\in\mathcal{C}_{\Lambda}$ such that $d(F,G)<\delta_1$. We thus have, for any $x\leq\overline{x}-\varepsilon$, \begin{equation*}
\Lambda(x)-G(x)\geq \Lambda(x)-F(x+\delta_1)-\delta_1 \geq \frac{1}{n_0}-\delta_1= \frac{1}{2n_0}>0.
\end{equation*}
It follows \begin{equation}\label{upper_cont}
\Lambda VaR(G)\leq \Lambda VaR(F)+\varepsilon
\end{equation}which is the upper semi-continuity.\\ By showing the lower semi-continuity we conclude the proof.
 From Definition \ref{defLVaR}, for any $\varepsilon>0$, there exists $\hat{x}\in [\overline{x},\overline{x}+\varepsilon]$ such that  $\gamma:=F(\hat{x})-\Lambda (\hat{x})>0$. Since $\Lambda$ is continuous, there exists $\delta>0$ such that for all $\delta'\leq \delta$, $\Lambda(\hat{x})-\Lambda(\hat{x}+\delta')\geq-\gamma/4$. Take now $\delta_2\leq \min\{\delta,\gamma/4,\varepsilon\}$ so that $\hat{x}+\delta_2\in [\overline{x},\overline{x}+\varepsilon]$. By observing that, for $G\in\mathcal{C}_{\Lambda}$ with $d(F,G)<\delta_2$ we have $$G(\hat{x}+\delta_2)-\Lambda(\hat{x}+\delta_2)\geq F(\hat{x})-\delta_2-\Lambda(\hat{x}+\delta_2)\geq F(\hat{x})-\Lambda(\hat{x})-\gamma/4-\delta_2 \geq \gamma/2 $$ we obtain 
 \begin{equation}\label{lower_cont}
\Lambda VaR(G)\geq -\hat{x}-\delta_2\geq \Lambda VaR(F)-\varepsilon.
\end{equation}
By taking $\delta:=\min\{\delta_1,\delta_2\}$ and combining \eqref{upper_cont} and \eqref{lower_cont}, we have that
$$\forall G\in \mathcal{C}_{\Lambda} \ \text{ with }d(F,G)<\delta \Longrightarrow|\Lambda VaR(F)-\Lambda VaR(G)|<\varepsilon$$
as desired.
 \end{proof}

The $\Lambda$ function adds flexibility to $\Lambda VaR$, however, when robustness is required, $\Lambda VaR$ should be constructed as suggested by the set $\mathcal{C}_{\Lambda}$. The $\Lambda$ function has to be chosen continuous and, on any interval, it cannot coincide with any distribution $F$ under consideration. We refer to Example \ref{EliRobLvaR} to show how this condition can be guaranteed given a set of normal distributions of P$\&$Ls.

\section{Elicitability}
\label{elicitability}
The importance of this property from a financial risk management perspective has been highlighted by \cite{EH14} as a consequence of the surprising results obtained by \cite{G11} and \cite{Z14}. Indeed, \cite{EH14} pointed out that the elicitability allows the assessment and the comparison of risk measure forecasting estimations and a straightforward backtesting.   

The term \emph{elicitable} has been introduced by \cite{L08} but the general notion dates back to the pioneering work of \cite{OR85}. In accordance with some parts of the literature, we introduce the notation $T:\mathcal{M}\subseteq  \mathcal{D} \rightarrow 2^{\mathbb{R}}$ to describe a set-valued statistical functional. Let us denote with $S(x,y)$ the realized forecasting error between the ex-ante prediction $x \in \mathbb{R}$ and the ex-post observation $y \in \mathbb{R}$, where $S$ is a function $S: \mathbb{R} \times \mathbb{R} \rightarrow [0, +\infty)$ called ``scoring'' or ``loss''. According to \cite{G11} a scoring function $S$ is consistent for the functional $T$ if
\begin{equation}\label{consistence} 
\mathbb{E}_F[S(t,Y)] \leq  \mathbb{E}_F[S(x,Y)] 
\end{equation}
for all $F$ in $\mathcal{M}$, all $t\in T(F)$ and all $x \in \mathbb{R}$.
It is strictly consistent if it is consistent and equality of the expectations implies that $x \in T(F)$.

\begin{definition}\citep{G11} \label{eli_G}
A set-valued statistical functional $T:\mathcal{M} \rightarrow 2^{\mathbb{R}}$ is elicitable if there exists a scoring function $S$ that is strictly consistent for it.
\end{definition}

\cite{BB14} have recently proposed a slightly different definition of elicitability. They consider only single-valued statistical functionals as a natural requirement in financial applications. In addition, they adopt additional properties for the scoring function. We also consider single-valued statistical functionals but without imposing any restriction on the scoring function. 
   
\begin{definition}\label{eli_argmin}
A statistical functional $T:\mathcal{M} \rightarrow \mathbb{R}$ is elicitable if there exists a scoring function $S$ such that 
\begin{equation}\label{argmin}
T(F)= \arg \min_{x} E_F[S(x,Y)] \quad \forall F \in \mathcal{M}.
\end{equation}
\end{definition}

Definition \ref{eli_G} restricted to the case of single-valued statistical functional is equivalent to Definition \ref{eli_argmin} when the minimum is unique. 
The statistical functional associated to a risk measure is the map $T:\mathcal{M}\rightarrow \overline{\mathbb{R}}$ such that $T(F)=-\rho(F)$ for any distribution $F$. We adopt this sign convention in accordance with part of the literature. We say that a risk measure is elicitable if the associated statistical functional $T$ is elicitable. In the following we will restrict to $\mathcal{M}\subseteq \mathcal{D}_1$ in order to have a finite expectation of the considered scoring functions.

The statistical functional associated to $VaR$, $T(F):=q_\lambda^+(F)$, is elicitable on the following set:
$$
\mathcal{M}_\lambda:=\left\{ F\in \mathcal{D}_1 : F \text{ strictly increasing } \right\}\subseteq\mathcal{C}_{\lambda}
$$
with $\mathcal{C}_{\lambda}$ as in \eqref{Clambda}, and with the following scoring function \citep{G11}:
\begin{equation}
S(x,y)=\lambda (y-x)^{+}+(1-\lambda )(y-x)^{-}.
\label{LossVaR}
\end{equation}

Let us denote with $T_\Lambda: \mathcal{D} \rightarrow \mathbb{R}$ the statistical functional associated to  $\Lambda VaR$ such that:
\begin{equation}\label{TLVaR} 
T_\Lambda(F)=-\Lambda VaR(F)
\end{equation}
and consider the set $\mathcal{M}_\Lambda \subseteq \mathcal{D}_{1}$ defined as follows:
\begin{equation}
\mathcal{M}_\Lambda=\{ F\in \mathcal{D}_1 : \exists \; \bar{x} \; \text{s.t. }\forall x<\bar{x} \text{, } F(x)< \Lambda(x) \text{ and } \forall x > \bar{x} \text{, } F(x)>\Lambda(x) \}.
\label{classeli}
\end{equation}
Once again this set coincides with $\mathcal{M}_\lambda$ when $\Lambda\equiv \lambda$. 
In \cite{BB14} it has been shown that $\Lambda VaR$ is elicitable under a stronger definition of elicitability and for the special case of $\Lambda$ continuous and decreasing. In the next theorem we prove that  $\Lambda VaR$ is elicitable using the general Definition \ref{eli_argmin} and under less restrictive conditions on $\Lambda$. Specifically, we show that $\Lambda VaR$ is elicitable on the particular class of distribution $\mathcal{M}_\Lambda$ in \eqref{classeli} depending on $\Lambda$.  
\begin{theorem}\label{ThmLVaR}
For any monotone and right continuous function $\Lambda: \mathbb{R} \rightarrow [\lambda^m, \lambda^M]$, with $0<\lambda^m \leq \lambda^M <1$, the statistical functional $T_\Lambda:\mathcal{D} \rightarrow \mathbb{R}$ defined in \eqref{TLVaR} is elicitable on the set $\mathcal{M}_\Lambda \subseteq \mathcal{D}_{1} $ defined in \eqref{classeli} with a loss function given by 
\begin{equation}
S(x,y)=(y-x)^- - \int_{y}^x \Lambda(t) dt.
\label{LossLVaR}
\end{equation}
\end{theorem}

\begin{proof}
We need to prove that $$T(F)=\text{arg}\min_x \int_{\mathbb{R}} S(x,y) dF(y).$$
In order to find a global minimum we first calculate the left and right derivatives of $\int _{\mathbb{R}} S(x,y) dF(y) $. Applying dominated convergence theorem we obtain:
\begin{align*}
\frac{\partial^{-}}{\partial x} \int _{\mathbb{R}} S(x,y) dF(y) &= \frac{\partial^{-}}{\partial x} \int _{\mathbb{R}} \Big( (y-x)^- - \int_{y}^x \Lambda(t)dt \Big) dF(y)  \\
&=\int_{\mathbb{R}} \Big( \frac{\partial^{-}}{\partial x} (y-x)^- -\frac{\partial^{-}}{\partial x} \int_{y}^x \Lambda(t)dt \Big) dF(y)\\
& = \int_{\mathbb{R}} \Big( \mathbf{1}_{(y<x)} - \Lambda(x^-) \Big)dF(y)\\
&=\lim_{t\uparrow x}F(t)-\Lambda(x^-)=F(x^-)-\Lambda(x^-).
\end{align*}
Analogously for the right derivative
\begin{align*}
\frac{\partial^+}{\partial x} \int _{\mathbb{R}} S(x,y) dF(y) & = \int_{\mathbb{R}} \Big( \mathbf{1}_{(y \leq x)} - \Lambda(x) \Big)dF(y)\\
&=F(x)-\Lambda(x).
\end{align*}
Observe now that $x^*=\inf\{x \in \mathbb{R} : F(x) > \Lambda(x)\}$, that is the statistical functional associated to $\Lambda VaR$, satisfies, for every $F\in\mathcal{M}_{\Lambda}$,
\begin{equation}\label{max_min}
 \begin{split}
 \forall x< x^* \hspace{0.4 cm} F(x) < \Lambda(x),&\hspace{0.4 cm} F(x^-) \leq \Lambda(x^-)\ ;\\
\forall x> x^* \hspace{0.4 cm} F(x) > \Lambda(x),&\hspace{0.4 cm} F(x^-) \geq \Lambda(x^-)\ ;
\end{split}
\end{equation}
from which we deduce 
\begin{equation}\label{max_min2}
 \begin{split}
\forall x<x^* \hspace{0.4 cm} &\frac{\partial^{-}}{\partial x} \int _{\mathbb{R}} S(x,y) dF(y) \leq 0 ,\hspace{0.4 cm} \frac{\partial^{+}}{\partial x} \int _{\mathbb{R}} S(x,y) dF(y) < 0 ;\\
\forall x>x^* \hspace{0.4 cm} &\frac{\partial^{-}}{\partial x} \int _{\mathbb{R}} S(x,y) dF(y) \geq 0 ,\hspace{0.4 cm} \frac{\partial^{+}}{\partial x} \int_{\mathbb{R}} S(x,y) dF(y) > 0 .
\end{split}
\end{equation}
This implies that $x^*$ is a local minimum. By showing that there are no other local minima we obtain that $x^*$ is the unique global minimum. Take first $x<x^*$. Observe that, by applying dominated convergence theorem,  $I(x):=\int_{\mathbb{R}} S(x,y) dF(y)$ is a continuous function. Moreover, $I$ is not constant on any interval in $(-\infty,x^*]$ since, from \eqref{max_min2}, we have $\frac{\partial^{+}}{\partial x} I < 0$. Since $I$ is continuous and, from \eqref{max_min2}, the left and right derivatives are non-positive, we have that any sequence converging to $x^-$ is decreasing. Analogously, any sequence converging to $x^+$ is increasing. In other words, there exists $\delta>0$ such that, $I(x_1)>I(x)>I(x_2)$ for all $x-\delta<x_1<x<x_2<x+\delta$. Thus $x$ is not a local minimum. The case $x>x^*$ is analogous. We can conclude that $\Lambda VaR$ is elicitable on the class of probability measures $\mathcal{M}_\Lambda$ defined in \eqref{classeli}.
\end{proof}

\begin{remark}
It is easy to prove that the scoring function in \eqref{LossLVaR} can be rewritten as follows:
\begin{equation}
S(x,y)=\frac{\int_{y}^x \Lambda(t) dt}{x-y}(y-x)^{+}+\left(1-\frac{\int_{y}^x \Lambda(t) dt}{x-y}\right)(y-x)^{-}.
\end{equation}
if $x\neq y$ and $S(x,x)=0$. It is evident the similarity with the scoring function of $VaR$ in \eqref{LossVaR}. Moreover, note that If $\Lambda$ is non-increasing obviously \eqref{classeli} is satisfied by every $F\in\mathcal{D}_1$ increasing so that we recover the result of \cite{BB14}.
\end{remark}

In general, the elicitability of $\Lambda VaR$ using the scoring function \eqref{LossLVaR} requires that $\Lambda$ is crossed only once by any possible $F$ at the level $\bar{x}=-\Lambda VaR(F)$ as shown in \eqref{classeli}.

\begin{remark}
$\Lambda VaR$ with a decreasing function $\Lambda$ is elicitable on the set of all the distributions. In this case, $\mathcal{M}_\Lambda \equiv \mathcal{D}_{1}$, since $F$ is non-decreasing and the derivatives of $\Lambda$ are negative. When $\Lambda$ is non-increasing, $\Lambda VaR$ is elicitable on the set of increasing distribution functions.
\end{remark}

If we additionally require continuity of $\Lambda$ we observe that $\mathcal{M}_\Lambda \subseteq \mathcal{C}_\Lambda$ where $\mathcal{C}_\Lambda$ is defined in \eqref{CLambda}. This implies that the set of distributions where $\Lambda VaR$ is elicitable guarantees also that $\Lambda VaR$ is robust. Hereafter, we provide an example of a construction of $\Lambda VaR$ with non-decreasing $\Lambda$ that is elicitable and robust given a set of normal distributions of P$\&$Ls. 

\begin{example}\label{EliRobLvaR}
Denote by $\Phi(x)$ the distribution function of a standard normal distribution. Let $\mathcal{M}:=\{\Phi(\frac{x-\mu_i}{\sigma_i})\}_{i\in I}$ for some collection $I$ such that $\overline{\mu}:=\sup \mu_i<\infty$ and $\underline{\sigma}:=\inf \sigma_i>0$. Set $\mu>\overline{\mu}$, $0<\sigma<\underline{\sigma}$ and define
\begin{equation*}
\Lambda(x):=\begin{cases}\lambda^m & x\leq x^m \\
\Phi\left(\dfrac{x-\mu}{\sigma}\right)& x^m \leq x <  x^M\\ \lambda^M & x\geq x^M. \\ \end{cases}
\end{equation*}
If $x^m\leq x^M$ are such that $0<\lambda_m\leq \Phi(\frac{x^m-\mu}{\sigma})$ and $\Phi(\frac{x^M-\mu}{\sigma})\leq \lambda^M$ then $\Lambda$ is non-decreasing and continuous. Moreover, from Theorem \ref{ThmLVaR}, $\Lambda VaR$ is elicitable on $\mathcal{M}$.
\end{example}

In order to have an elicitable $\Lambda VaR$ with the scoring function \eqref{LossLVaR} we need to build the $\Lambda$ function under a certain condition that depends on the set of the P$\&$L distributions. In particular, the scoring function \eqref{LossLVaR} guarantees the elicitability of $\Lambda VaR$ with non-decreasing $\Lambda$ only in the class of probability measures $\mathcal{M}_\Lambda$ in \eqref{classeli} as shown by the following counterexample.

\begin{example}
 Let $\varepsilon<0.5\%$. Let $\Lambda(x)$ and $F(x)$ as follows
\begin{equation*}F(x)=\left\{\begin{array}{ll}
0& x<-100\\
1.5\% &-100\leq x< 4\\
1& x\geq 4\\
\end{array}\right.\qquad
\Lambda(x)=\varepsilon+\left\{\begin{array}{ll}
0& x<-101\\
(x+101)/100 &-101\leq x< -99\\
2\% & x\geq -99.\\
\end{array}\right.
\end{equation*}
 
$F(x)$ is the cumulative distribution function of a random variable $Y$ with distribution: $Y=-100$ with probability $p=1.5\%$ and $Y=4$ with probability $1-p=98.5\%$.

It is easy to compute that the statistical functional associated to $\Lambda VaR$ is $T_{\Lambda}(F)=-100$. If $\Lambda VaR$ is elicitable $T_{\Lambda}$ should be the minimizer of $$g(x):=\mathbb{E}[S(x,Y)]=S(x,-100)\dfrac{1.5}{100}+S(x,4)\dfrac{98.5}{100}.$$ Since $S$ for $\Lambda VaR$ is defined as in \eqref{LossLVaR}, we need compute the primitive for $\Lambda$ that is given by
$$\Psi(t)=\int \Lambda(t)=\varepsilon t+\left\{\begin{array}{ll}
0& t<-101\\
\dfrac{(t^2/2+101t)}{100} &-101\leq t< -99\\
\dfrac{2}{100}t & t\geq -99.\\
\end{array}\right.$$
Hence, $\Psi(-100)=-51-100\varepsilon$ and $\Psi(4)=8/100+4\varepsilon$, thus, we have $S(x,-100)=(-100-x)^--\Psi(x)-51-100\varepsilon$ and $S(x,4)=(4-x)^--\Psi(x)+8/100+4\varepsilon$ and
$$g(x)=-\Psi(x)+(-100-x)^-\dfrac{1.5}{100}+(4-x)^-\dfrac{98.5}{100}+c$$
where $c=(-51-100\varepsilon)\cdot1.5\%+(0.08+4\varepsilon)\cdot98.5\%$.
Observe now that $\Lambda VaR$ is not the global minimum, since $g(-100)>g(4)$. Indeed:
$$g(-100)-g(4)=-\Psi(-100)+\Psi(4)-104\cdot\dfrac{1.5}{100}=51+\dfrac{8}{100}-104\cdot\dfrac{1.5}{100}>0.$$
\end{example}

\bigskip

We have shown that the scoring function in \eqref{LossLVaR} guarantees elicitability of $\Lambda VaR$ only on the set of distributions $\mathcal{M}_\Lambda$ in \eqref{classeli}. Whether there exists another scoring function that guarantees the elicitability of $\Lambda VaR$ on a larger class of distributions is an interesting question which might be object of further studies. We conclude this Section by discussing some insights on this problem and the difficulties that might arise for such an extension. In particular we investigate a necessary condition for elicitability, namely, the convex level sets property \citep{OR85}. 
\begin{definition}
 If $\mathcal{M}\subseteq \mathcal{D}$ is convex we say that $T$ has convex level sets if, for any $\gamma \in \mathbb{R}$, the level sets
\begin{equation*}
\{T=\gamma \}:=\{F \in \mathcal{M} : T(F) =\gamma \} 
\end{equation*} 
are convex, i.e. for any $\alpha \in [0,1]$ and $F_1,F_2 \in \mathcal{M}$
\begin{equation*}
T(F_1)=T(F_2)=\gamma \Rightarrow T(\alpha F_1+(1-\alpha)F_2)=\gamma.
\end{equation*} 
\end{definition}

\begin{proposition}
\citep{OR85} If a statistical functional $T:\mathcal{M}\subseteq \mathcal{D} \rightarrow \mathbb{R}$ is elicitable, then $T$ has convex level sets.
\end{proposition}

\cite{G11} showed that $ES$ does not satisfy this necessary condition, as a consequence, $ES$ is not elicitable. 

We have shown in Theorem \ref{ThmLVaR} that $\Lambda VaR$ is elicitable in $\mathcal{M}_\Lambda$, hence, it also has convex level sets in this class of distributions. The following example shows that, in general, $\Lambda VaR$ might not satisfy this condition on a larger set of distributions and, thus, neither elicitability.
\begin{example} Fix $0<\varepsilon<\frac{1}{2}$ and $\lambda^M<1$. Consider 
\begin{equation*}
F_1(x):=\sum_{k=1}^{\infty}\dfrac{1}{2^k}\mathbf{1}_{\left[\frac{1}{k+1},\frac{1}{k}\right)}(x)+\varepsilon\mathbf{1}_{[0,1)} +\mathbf{1}_{[1,\infty)}
\end{equation*} and
\begin{equation*}
F_2(x):=F_1(x)+\sum_{k=1}^{\infty}(-1)^{k}\dfrac{1}{10^{k}}\mathbf{1}_{\left[\frac{1}{k+1},\frac{1}{k}\right)}(x).
\end{equation*}
As a function $\Lambda$ take $\Lambda:=\varepsilon\mathbf{1}_{(-\infty,0)}+\frac{1}{2}(F_1+F_2)\mathbf{1}_{[0,1)}+\lambda^M\mathbf{1}_{[1,\infty)}$. Observe that $\forall k\in \mathbb{N}$ $F_1(\frac{1}{2k})>\Lambda(\frac{1}{2k})$ and $F_2(\frac{1}{2k+1})>\Lambda(\frac{1}{2k+1})$. Moreover, $0=F_1(x)=F_2(x)<\Lambda(x)$ for all $x< 0$. This implies $\Lambda VaR(F_1)=\Lambda Var(F_2)=0$. Nevertheless, since $\Lambda(x)=\lambda^M<1$ for $x\geq 1$, we have $\Lambda Var(\frac{1}{2}F_1+\frac{1}{2}F_2)=-1$, from which the convex level set property fails.
\end{example}
A positive answer for the convex level sets property is given by the choice of a particular class of $\Lambda$ for which the condition is satisfied on the set of increasing distribution functions. 
\begin{lemma} 
If $\Lambda$ is non-decreasing and piecewise constant with a finite number of jumps, then $\Lambda VaR$ has convex level sets on the set of increasing distribution functions.
\end{lemma}
\begin{proof}
We first observe that, in general, $T_{\Lambda}(F_i)=\gamma$ for $i=1,2$ implies $T_\Lambda(\alpha F_1+(1-\alpha )F_2)\geq \gamma$ for every  $\alpha \in [0,1]$ and $F_1,F_2 \in \mathcal{D}$. To this end, we prove that $\inf\{x \; :\alpha F_{1}(x)+(1-\alpha)F_{2}(x) >  \Lambda(x)\} \geq \gamma$, with $\gamma=T_\Lambda(F_i):=\inf \{x\; :F_i(x) > \Lambda(x)\}$ for $i=1,2$. Note that by definition of $T_{\Lambda}(F_i)$ for $i=1,2$, we have $F_i(x)\leq \Lambda(x)$ for every $x\leq \gamma$. We thus get, for an arbitrary $0\leq\alpha\leq 1$, $\alpha F_{1}(x)+(1-\alpha)F_{2}(x)\leq \Lambda(x)$ for every $x\leq \gamma$ from which $T_{\Lambda}(\alpha F_1+(1-\alpha) F_2)\geq \gamma$.

For the converse inequality observe that there exists $\varepsilon>0$ such that $\Lambda$ is constant on $[\gamma,\gamma+\varepsilon)$. Since $\gamma=\inf \{x\; :F_i(x) > \Lambda(x)\}$  and $F_i$ is non-decreasing, for $i=1,2$, then $\alpha F_1(x)+(1-\alpha) F_2(x)>\Lambda(x)$ for every $x\in(\gamma,\gamma+\varepsilon)$  from which $T_\Lambda(\alpha F_1+(1-\alpha )F_2)\leq \gamma$.
\end{proof}

In conclusion, we have observed that extending the class of distributions for which the convex level sets property holds depends heavily on the specific choice of $\Lambda$ and hence it seems to necessitate a case-by-case study.

\section{Consistency} 
\label{consistency}
In this section we refer to the notion of \textit{consistency} recently studied by \cite{Davis}. Davis recognized the importance of the elicitability property in the backtesting context of risk measures, but he argued that the problem can be better addressed from a different perspective.
The motivation of Davis' study relays on the difficulties of predicting the ``true'' distribution $F$ of portfolio financial returns. Suppose indeed you are given the information up to time $k-1$, at time $k$ only one realization occurs and so there is not enough information to claim if the prediction of $F$ was correct or not. 
Thus, Davis introduces the notion of consistency of a risk estimator that is based on the daily comparison between the realization of the risk estimator and the realized outcome, but without consideration how the predictions were arrived at. Hence, the fundamental difference with the elicitability property is that the assumption on the model generating the conditional distribution of the portfolio returns can change at any time and one should just check if the prediction is performing well or not \citep[see][for a more detailed discussion]{Davis}. \\

In this section we adopt the framework of Davis. Namely, we fix $(\Omega,\mathcal{F},\{\mathcal{F}_k\}_{k \in \mathbb{N}})$ where $\Omega=\prod_{k=1}^{\infty}\mathbb{R}_{(k)}$ is the canonical space for a real-valued data process $Y=\{Y_k\}_{k\in\mathbb{N}}$; $\mathcal{F}$ is the product sigma-algebra generated by the Borel sigma-algebra in each copy of $\mathbb{R}$ (denoted by $\mathbb{R}_{(k)}$); $\{\mathcal{F}_k\}_{k \in \mathbb{N}}$ is the natural filtration of the process $Y$ and $\mathcal{F}_0$ the trivial sigma-algebra. The class of possible models, for this data process, is represented by a collection $\mathcal{P}$ of probability measures denoted by $\mathcal{P}:=\{\mathbb{P}^{\alpha}, \alpha \in \mathfrak{A}\}$, where $\mathfrak{A}$ is an arbitrary index set. We denote  with $\mathbb{E}^{\alpha}$ the expectation with respect to $\mathbb{P}^{\alpha}$. For every $\mathbb{P}^{\alpha}$ it is possible to define, for each $k \geq 1$, the conditional distribution of the random variable $Y_k$ given $\mathcal{F}_{k-1}$, as a map $F^{\alpha}_k:\mathbb{R}\times\Omega\mapsto[0,1]$ satisfying: for $\mathbb{P}^\alpha$-a.e. $\omega$, $F^{\alpha}_k(\cdot,\omega)$ is a distribution function, and for every $x\in\mathbb{R}$, $F^{\alpha}_k(x)=\mathbb{P}^{\alpha}(Y_k \leq x | \mathcal{F}_{k-1})$ $\mathbb{P}^{\alpha}$-a.s.  

\begin{definition}\citep{Davis} \label{defConsistenza}
Let $\mathfrak{B}(\mathcal{P})$ be a set of strictly increasing predictable processes $b=\{b_n\}_{n \in \mathbb{N}}$ such that $ \lim_{n \rightarrow \infty}b_n= \infty$ $\mathbb{P}^{\alpha}$-a.s. for every $\alpha\in\mathfrak{A}$, and $l: \mathbb{R}^2 \rightarrow \mathbb{R}$ a calibration function, that is a measurable function such that $\mathbb{E}^{\alpha}[l(T(F^{\alpha}_k),Y_k)|\mathcal{F}_{k-1}]=0$ for all $\mathbb{P}^{\alpha} \in \mathcal{P}$ . A risk measure $\rho$ is $(l,b, \mathcal{P})$-consistent if the associated statistical functional $T$ satisfies
\begin{equation} \label{defCons}
\lim_{n \rightarrow \infty} \frac{1}{b_n} \sum_{k=1}^n l(T(F^{\alpha}_k),Y_k)=0 \quad \mathbb{P}^\alpha\text{-a.s. } \forall \mathbb{P}^\alpha \in \mathcal{P}.
\end{equation}
\end{definition}Denote by $\mathfrak{P}$ the set of all probability measures and define: 
$$\mathcal{P}^0=\{\mathbb{P}^\alpha \in \mathfrak{P} : \forall k \;  F^\alpha_k(x, \omega) \; \text{ is continuous in} \;  x \; \text{for }\mathbb{P}^\alpha\text{-almost all} \; \omega \in \Omega\}.$$\cite{Davis} showed that $VaR$ satisfies this consistency property for a large class of processes $\mathcal{B}(\mathcal{P})$ and for the large class of data models $\mathcal{P}^0$ with the following calibration function:
$$l(x,y)=\lambda -\mathbf{1}_{(y\leq x)}.$$
The statistical functional associated to $\Lambda VaR$ is given by \eqref{TLVaR}, hence we define for every $k$ and $\alpha\in\mathfrak{A}$:
$$T_\Lambda(F^{\alpha}_k):=\inf\{x\mid F^{\alpha}_k(x)> \Lambda(x)\}.$$
Notice that $\{T_\Lambda(F^{\alpha}_k)\}_{k\in\mathbb{N}}$ and 
$\{\Lambda(T_{\Lambda}(F^{\alpha}_k))\}_{k\in\mathbb{N}}$ are predictable process, as shown in the following lemma.  
\begin{lemma}\label{LemmaMeas}For every $k\geq 1$,
$T_\Lambda(F^{\alpha}_k)$ and $\Lambda(T_\Lambda(F^{\alpha}_k))$ are $\mathcal{F}_{k-1}$-measurable random variables. 
\end{lemma}

\begin{proof}
Fix a probability $\mathbb{P}^\alpha$ with $\alpha\in\mathfrak{A}$. Notice first that for any $y\in\mathbb{R}$, for $\mathbb{P}^\alpha$ a.e. $\omega$, we have
\begin{equation*}
\begin{array}{cccc}
T_\Lambda(F^{\alpha}_k)\geq y&\Longleftrightarrow & F^{\alpha}_k(x)\leq \Lambda(x) & \forall x\leq y\\
&\Longleftrightarrow & F^{\alpha}_k(q)\leq \Lambda(q) & \forall q\in\mathbb{Q},\ q\leq y
\end{array}
\end{equation*}
where the last equivalence follows from the right-continuity of $F^\alpha_k$ and $\Lambda$. We therefore have
$$\{\omega\mid T_\Lambda(F^{\alpha}_k)\geq y\}=\bigcap_{q\in\mathbb{Q}\cap (-\infty,y]}\{\omega\mid F^{\alpha}_k(q)\leq \Lambda(q)\}\in\mathcal{F}_{k-1}$$
from which $T_{\Lambda}(F_k^\alpha)$ is an $\mathcal{F}_{k-1}$-measurable random variable. 

$\Lambda(T_{\Lambda}(F_k^\alpha))$ is also $\mathcal{F}_{k-1}$-measurable: since $\Lambda$ is right-continuous $\Lambda(x)\geq y$ iff 
$x\geq\Lambda^{-}(y)$ where $\Lambda^{-}(y):=\inf\{x\in\mathbb{R}\mid \Lambda(x)\geq y\}$ is the generalized inverse \citep[Proposition 1]{EH13} and thus
$$\{\omega\mid \Lambda(T_\Lambda(F^{\alpha}_k))\geq y\}=\{\omega\mid T_\Lambda(F^{\alpha}_k)\geq \Lambda^-(y)\}\in\mathcal{F}_{k-1}.$$
\end{proof}

By following the methodology suggested by \cite{Davis}, we are able to show that $\Lambda VaR$ is consistent for the large class of data models $\mathcal{P}^0$, as shown in the following theorem.
 
\begin{theorem}\label{thmConsistency}
For each $\mathbb{P}^\alpha\in \mathcal{P}^0$, 
\begin{equation} \label{cons_LambdaVaR}
\frac{1}{n}\sum_{k=1}^n \Lambda(T_\Lambda(F^{\alpha}_k))-\mathbf{1}_{(Y_k \leq T_\Lambda(F^{\alpha}_k))} \rightarrow 0 \quad \mathbb{P}^\alpha\text{-a.s.}
\end{equation} \label{cal_LambdaVaR}
Thus, $\Lambda VaR$ is $(l,n,\mathcal{P}^0)$-consistent with 
\begin{equation}l(x,y)=\Lambda(x)-\mathbf{1}_{(y \leq x)}.
\end{equation}
\end{theorem}
Before giving the proof of the theorem we show the following lemma. 
\begin{lemma}\label{LemmaCond}For each $\mathbb{P}^\alpha\in \mathcal{P}^0$,
$$\mathbb{E}^{\alpha}\left[\mathbf{1}_{(Y_k\leq T_\Lambda(F^{\alpha}_k))}\mid\mathcal{F}_{k-1}\right]=\Lambda(T_\Lambda(F^{\alpha}_k)),\qquad \mathbb{P}^\alpha\text{-a.s.}$$
\end{lemma}
\begin{proof}Fix $\mathbb{P}^\alpha\in \mathcal{P}^0$. Since there is no confusion, for ease of notation, we omit the dependence on $\alpha$.
Observe that $U_k:=F_k(Y_k)$ is uniformly distributed and $$Y_k\leq T_\Lambda(F_k) \Longleftrightarrow U_k\leq F_k(T_\Lambda(F_k))=\Lambda(T_\Lambda(F_k)).$$
Note now that $U_k$ is independent of $\mathcal{F}_{k-1}$ since, from the continuity of $F_k$, $\mathbb{P}(U_k\leq u_k\mid \mathcal{F}_{k-1})=\mathbb{P}(Y_k\leq F_k^-(u_{k})\mid\mathcal{F}_{k-1})=F_k(F_k^-(u_k))=u_{k}=\mathbb{P}(U_k\leq u_k)$ (where $F_k^-$ denotes the generalized inverse of $F_k$). Since $\Lambda(T_\Lambda(F_k))$ is $\mathcal{F}_{k-1}$-measurable from Lemma \ref{LemmaMeas}, we can compute the desired conditional expectation through the application of the freezing lemma \citep[Section 9.10]{Williams}. Namely, define $h(x,y):=\mathbf{1}_{\{y\leq x\}}$ and let $\hat{h}(x):=\mathbb{E}[\mathbf{1}_{\{U_k\leq x\}}]=x$.  Since $h$ is a bounded Borel-measurable function and $U_k$ is independent of $\mathcal{F}_{k-1}$, then 
\begin{equation*}
\mathbb{E}\left[\mathbf{1}_{(Y_k\leq T_\Lambda(F_k))}\mid\mathcal{F}_{k-1}\right]=\hat{h}(\Lambda(T_\Lambda(F_k)))=\Lambda(T_\Lambda(F_k))
\end{equation*}
where equalities are intended in the $\mathbb{P}$-a.s. sense.
\end{proof}

\bigskip

\begin{proof}[Proof of Theorem \ref{thmConsistency}]
Define:
\begin{equation*}
Z_k:=\Lambda(T_\Lambda(F^{\alpha}_k))-\mathbf{1}_{(Y_k \leq T_\Lambda(F^{\alpha}_k))},
\end{equation*}
\begin{equation*}
S_n:=\sum_{k=1}^n Z_k,\qquad Q_n:=\sum_{k=1}^n (Z_k)^2,\qquad \langle S\rangle_n:=\sum_{k=1}^n \mathbb{E}^\alpha[(Z_k)^2\mid\mathcal{F}_{k-1}].
\end{equation*}
From Lemma \ref{LemmaCond}, $S_n$ is a martingale since $\mathbb{E}^\alpha[S_n-S_{n-1}\mid\mathcal{F}_{n-1}]=\mathbb{E}^\alpha[Z_n\mid\mathcal{F}_{n-1}]=0$. We now compute $(Z_k)^2$, we use the shorthand $W:=\Lambda(T_\Lambda(F^{\alpha}_k))$ .
\begin{eqnarray*}
(Z_k)^2&=&\mathbf{1}_{(Y_k \leq T_\Lambda(F^{\alpha}_k))}+W^2-2W\mathbf{1}_{(Y_k \leq T_\Lambda(F^{\alpha}_k))}\\
&=& W^2\mathbf{1}_{(Y_k >T_\Lambda(F^{\alpha}_k))}+(1-W)^2\mathbf{1}_{(Y_k \leq T_\Lambda(F^{\alpha}_k))}.
\end{eqnarray*}
Note that since $\lambda^m\leq W\leq \lambda^M$ we obtain $\mathbb{E}^\alpha[(Z_k)^2]\leq\max\{(\lambda^M)^2,(1-\lambda^m)^2\}<\infty$ so that $S_n$ is a square integrable martingale. Moreover, observe that, 
\begin{equation}\label{lowerBound}
(Z_k)^2\geq\min\{(\lambda^m)^2,(1-\lambda^M)^2\}.
\end{equation}
Since $W$ is $\mathcal{F}_{k-1}$-measurable, using Lemma \ref{LemmaCond},
\begin{eqnarray*}
\mathbb{E}^\alpha[(Z_k)^2\mid\mathcal{F}_{k-1}]&=& W^2\mathbb{E}^\alpha[\mathbf{1}_{(Y_k > T_\Lambda(F^{\alpha}_k))}\mid\mathcal{F}_{k-1}]+(1-W)^2\mathbb{E}^\alpha[\mathbf{1}_{(Y_k \leq T_\Lambda(F^{\alpha}_k))}\mid\mathcal{F}_{k-1}]\\&=& W^2(1-W)+(1-W)^2 W\\
&=& W(1-W).
\end{eqnarray*}
It follows that $$\lambda^m(1-\lambda^M)\leq \mathbb{E}^\alpha[(Z_k)^2\mid\mathcal{F}_{k-1}]\leq \lambda^M(1-\lambda^m)$$which firstly implies $\langle S\rangle_n\geq n\lambda^m(1-\lambda^M)\rightarrow\infty$, and, secondly, combined with \eqref{lowerBound},
\begin{equation*}
\dfrac{Q_n}{\langle S\rangle_n}\geq \dfrac{n\min\{(\lambda^m)^2,(1-\lambda^M)^2\}}{n\lambda^M(1-\lambda^m)}=\dfrac{\min\{(\lambda^m)^2,(1-\lambda^M)^2\}}{\lambda^M(1-\lambda^m)}=:\varepsilon_\alpha>0.
\end{equation*}
Notice that $Z_k$ is bounded from above by 1 for every $k \in \mathbb{N}$, thus, we have  $Q_n \leq n$.
We can therefore conclude
\begin{equation*}
\left|\dfrac{S_n}{n}\right|\leq \left|\dfrac{S_n}{Q_n}\right|=\dfrac{\langle S\rangle_n}{Q_n}\left|\dfrac{S_n}{\langle S\rangle_n}\right|\leq \dfrac{1}{\varepsilon_{\alpha}}\left|\dfrac{S_n}{\langle S\rangle_n}\right|\rightarrow 0  \quad \mathbb{P}^\alpha\text{-a.s}.
\end{equation*}
where the last term converges to $0$ from Proposition 6.3 in \cite{Davis}.
\end{proof}

\bigskip

As a consequence of the theorem, similarly to what observed by \cite{Davis} for $VaR$, a risk manager could use the following relative frequency measure 
\begin{equation} 
\frac{1}{n}\sum_{k=1}^n\Lambda(T_\Lambda(F^{\alpha}_k))-\mathbf{1}_{(Y_k \leq T_\Lambda(F^{\alpha}_k))} 
\label{statisticalvar}
\end{equation} 
as test statistic in a finite-sample hypothesis test \citep[as considered in][]{CP}. Obviously $\Lambda VaR$ is also $(l, b',\mathcal{P}^0)$-consistent with $b'_n=n b_n$ and $b=\{b_n\}_n \in \mathfrak{B}(\mathcal{P})$.

Therefore, the consistency of $\Lambda VaR$, as the quantile forecasting, can be obtained under essentially no conditions on the mechanism generating the data. This is not the case of the estimates of the statistical functional $T_{m}$ associated to the conditional mean (such as $ES$) and defined as follows:
$$T_{m}(F^{\alpha}_k):= \int_{\mathbb{R}} x F^{\alpha}_k(dx).$$
Indeed, \cite{Davis} showed that $T_m(F^{\alpha}_k)$ satisfies the condition \eqref{defCons}  with
$l(x,y)=x-y$, $Q_n=\sum_{k=1}^n Z_k^2$, where $Z_k:=Y_k-T_{m}(F_k^\alpha)$, and, remarkably,   
$\mathcal{P}^1 \in \mathfrak{P}$ is the set of probability measures such that: 
\begin{itemize}
\item[i)] for any $k$, $Y_k \in L^2(\mathbb{P}^\alpha)$,
\item[ii)] $\lim_{n \rightarrow \infty} \langle S \rangle_n=\infty$  $\mathbb{P}^\alpha$-a.s., with $\langle S \rangle_n:=\sum_{k=1}^n \mathbb{E}[Z_k^2|\mathcal{F}_{k-1}]$,
\item[iii)] there exists $\varepsilon_{\alpha}>0$ such that $\frac{Q_n}{\langle S \rangle_n}> \varepsilon_{\alpha}$ for large $n$, $\mathbb{P}^{\alpha}$-a.s.
\end{itemize}

In general, the validity of conditions i), ii), iii) might be difficult to check. In addition, the process $Q_n$ is not predictable, thus, it is not possible to conclude that statistical functionals that depends on the mean (such as $ES$) satisfy the consistency property as in Definition \ref{defConsistenza} using this methodology \footnote{We thank an anonymous referee  that pointed out this issue.}. Hence, in line with the elicitability framework, verifying the accuracy of mean-based estimates is definitely more problematic than the same problem for quantile-based forecasts. For the case of $\Lambda VaR$ this is possible and all the conditions are satisfied so that the methodology can be successfully applied.

\section{Conclusions}
We have shown that $\Lambda VaR$, satisfies robustness and elicitability in particular classes of distributions. Robustness requires that the $\Lambda$ function is continuous and does not coincide with the distribution $F$ on any interval. Elicitability requires a bit more, that is, $\Lambda$ is crossed only once by any possible $F$. We have also proposed an example of construction of an elicitable and robust $\Lambda VaR$ given a set of normal distributions. In addition, we have shown that $\Lambda VaR$ satisfies the consistency property without any conditions on the mechanism generating data, allowing a straightforward back-testing.

After the recent financial crisis, the \cite{FRTB} has suggested that banks should abandon $VaR$ in favour of the $ES$ as a standard tool for risk management since $ES$ is able to overcome two main shortcomings of $VaR$: lack of convexity on random variables and insensitivity with respect to tail behaviour. However, $ES$ has also some issues. Specifically, $ES$ is not robust, or only for small degrees when a stronger definition of robustness is required, and it is not elicitable. Recently, \cite{AS14} showed that the elicitability of $ES$ can be reached jointly with $VaR$ \citep[see also][for an extended result]{FZ14}. In addition,  verifying the consistency property for  $ES$ is more problematic. Moreover, a  recent study by \cite{KM15} pointed out that not all the aspects of  $ES$ are well understood. For instance, for positions with a high probability of losses but also high expected gains in the tails, $ES$ does not necessarily perform better than $VaR$ from a liability holders' perspective. Other risk measures which consider the magnitude of losses beyond $ES$ are the expectiles, recently studied by \cite{BD15}.

 In any case, the issue of capturing tail risk remains crucial and cannot be accomplished through $VaR$. The new risk measure, $\Lambda VaR$, may solve this issue since it is able to discriminate the risk among distributions with the same quantile but different tail behaviour and shares with $VaR$ other important properties such as quasi-convexity. On the other hand, $\Lambda VaR$ lacks subadditivity and the flexibility introduced by the $\Lambda$ function requires additional criteria for determining its upper and lower bound. However, we think that $\Lambda VaR$ may be considered as an alternative risk measure valuable for further studies. 

\section*{Acknowledgement}
We wish to thank the two anonymous referees for useful comments and J. Corbetta for helpful discussions on this subject.\\
ETH foundation is gratefully acknowledge for supporting this research.

 \end{document}